\documentclass{article}
\usepackage[latin1]{inputenc}

\usepackage{amsmath,amssymb,amsbsy}
\usepackage{amsthm}
\usepackage[english]{babel}
\usepackage{graphicx}
\usepackage{subfigure}
\usepackage{natbib}

\renewcommand{\operatorname}[1]{\mathop{\mathrm{#1}}}
\newcommand{\tr}{\mathrm{\mathop{tr}}}
\newcommand{\T}{^\mathrm{T}}
\newcommand{\Real}{\mathbb{R}}

\newtheorem{theorem}{Theorem}
\theoremstyle{remark}
\newtheorem{alg}{Algorithm}

\usepackage{tikz}
\usetikzlibrary{arrows}
\usetikzlibrary{shapes}
\usetikzlibrary{plotmarks}
\usetikzlibrary{backgrounds}

\tikzstyle{graph}=[line width=1.2pt,join=round]

\title{Global Sensitivity Analysis of Biochemical Reaction Networks via Semidefinite Programming}
\author{Steffen Waldherr\footnotemark[1] \and Rolf Findeisen\footnotemark[2] \and Frank Allgöwer\footnotemark[1]}
\date{\footnotemark[1] Institute for Systems Theory and Automatic Control,\\ Universität Stuttgart, Germany \\
\footnotemark[2] Chair for Systems Theory and Automatic Control, Otto-von-Guericke-Universität Magdeburg, Germany}

\begin{document}
%
%


\maketitle

\begin{abstract}
We study the problem of computing outer bounds 
for the region of steady states of biochemical
reaction networks modelled by ordinary differential equations,
with respect to parameters that are allowed to vary within a predefined region.
Using a relaxed version of the corresponding feasibility problem and its
Lagrangian dual, we show how to compute certificates for regions in state
space not containing any steady states.
Based on these results, we develop an algorithm to compute outer bounds for the region
of all feasible steady states.
We apply our algorithm to the sensitivity analysis of a Goldbeter--Koshland enzymatic
cycle, which is a frequent motif in reaction networks for regulation of metabolism and signal
transduction.
{\it Copyright \copyright~2008 IFAC}.
\end{abstract}

\section{Introduction}
\label{sec:intro}

A basic question in the analysis of biochemical reaction networks is how steady state
concentrations change with parameters.
Metabolic Control Analysis (MCA) is a classical tool to answer this question \citep{KacserBur1995},
where the analysis is based on a linear approximation of the system's equations around the
steady state.
Due to the linear approximation, results from MCA are only valid if parameter variations
are small.
However, in natural biochemical reaction networks, one usually faces large parameter variations:
in genetic engineering, common techniques like gene knock-outs or knock-downs, overexpression
or binding site mutations typically give rise to large parameter variations.

It follows that there is a need to compute changes in steady state values 
which are due to large parameter variations.
One approach to broaden the validity of results from MCA to larger parameter variations
is to include higher order approximations at the nominal point \citep{StreifFin2007}.
Although such an approach may extend the validity of the approximation, it still gives
results which are in general only locally valid.

We will thus rather take a different route and study the problem from the perspective
of computing the set of all steady states
for given ranges in which parameter values may vary.
In contrast to classical, local sensitivity analysis, such an approach allows to directly evaluate
the range that steady state concentrations can take for given parameter ranges.
The drawback is that it is not directly possible to assess the influence of individual
parameters on the steady state.
However, by repeating the computation for different parameter ranges, also this
information may be obtained.

Computing the set of steady states analytically is only possible in very rare cases. 
Even if an analytical solution
for the steady state is known, computing the corresponding set for all possible
parameter values may be difficult.
Due to this difficulty, non-deterministic approaches are frequently used to solve this problem.
A common tool for this kind of analysis are Monte Carlo methods \citep{RobertCas2004},
which are routinely applied in the analysis of uncertain biochemical reaction 
networks \citep{AlvesSav2000,FengHoo2004}.
However, Monte Carlo methods do not give reliable results in the sense that it is possible to
miss important solutions, which is particularly problematic for highly nonlinear
dependencies of the steady state on parameters.
Also, Monte Carlo approaches to the problem at hand typically require that all 
of the possibly multiple steady states for
specific parameter values can be computed explicitly, which is often a difficult task
in itself.

Continuation methods that track the changes in steady state values upon parameter
variations are an efficient computational tool for
this problem \citep{RichterDeC1983,Kuznetsov1995},
but are restricted to low-dimensional parameter variations
and are thus in general unsuitable for exploring higher-dimensional parameter spaces.

Global optimization methods employing branch and bound techniques or interval arithmetics
would in principle be suited to compute steady state regions \citep{MaranasFlo1995,Neumaier1990}.
However, it seems that the corresponding computational cost has obstructed their application to
the analysis of biochemical reaction networks so far.

In this paper, we propose a new approach to obtain reliable bounds on steady state
values under uncertain parameters in a computationally efficient way.
The paper is structured as follows.
In Section~\ref{sec:feasibility}, we study the problem of finding certificates that a given
set in state space does not contain a steady state for any parameters in
a given set in parameter space.
In Section~\ref{sec:bounding}, we use the results obtained in Section~\ref{sec:feasibility}
to develop an algorithm that computes outer bounding regions of steady state values for 
a given set in which parameters vary.
The application of the proposed analysis method is shown for two example reaction networks
in Section~\ref{sec:examples}.

\subsection*{Mathematical notation}

The space of real symmetric $n\times n$ matrices is denoted as $\mathcal S^n$.
The order operator with respect to the positive orthant in $\Real^{m\times n}$ is denoted as
``$\boldsymbol\leq$'',
i.e.\ $0\boldsymbol\leq X\in\Real^{m\times n} \Leftrightarrow 0 \leq X_{i,j}$ for $i=1,\dotsc,m$, $j=1,\dotsc,n$.
The order operator with respect to the cone of positive semidefinite (PSD) matrices
in $\mathcal S^n$ is denoted as ``$\preccurlyeq$'', i.e.\ $0\preccurlyeq X\in\mathcal S^n \Leftrightarrow
X$ is PSD.
The trace of a quadratic matrix $X\in\Real^{n\times n}$ is denoted as $\tr X$.

\section{Problem statement and basic idea}

We consider biochemical reaction networks that are modelled by ordinary differential equations.
This modelling framework is quite general and covers most metabolic networks as well as
many signal transduction pathways, if spatial effects can be neglected.
Mathematically, such models are commonly written as
\begin{equation}
\label{eq:ode}
\begin{aligned}
\dot x = Sv(x,p),
\end{aligned}
\end{equation}
where $x\in\Real^n$ is the concentration vector, $S\in\Real^{n\times m}$ is the stoichiometric
matrix, $p\in\Real^m$ is the vector of parameter values and $v(x,p)\in\Real^m$ is the
vector of reaction fluxes \citep{KlippHer2005}.
Throughout this paper, we assume that fluxes are modelled using the law of mass action,
where $v$ takes the form
\begin{equation}
\label{eq:flux}
\begin{aligned}
v_j(x,p) = p_j \prod_{k=1}^n x_k^{\sigma_{jk}},
\end{aligned}
\end{equation}
for $j=1,\dotsc,m$.
The constants $\sigma_{jk}$ are integers representing the stoichiometric coefficient
of the species $k$ taking part in the $j$-th reacting complex.
In the case of mass action kinetics, the dimensions of the parameter vector and the 
flux vector are in general the same.
Note that our results can be extended to rational functions describing the fluxes,
such as used for Michaelis--Menten kinetics, in a straightforward way.

The problem under consideration can be formulated as follows.
Given a set
$\mathcal P \subset \Real^m$ in parameter space, compute a set $\mathcal X_s \subset \Real^n$
that contains all steady states of the system~\eqref{eq:ode} for parameter
values taken from $\mathcal P$.
Ideally, the set $\mathcal X_s$ should be as small as possible, 
such that for all $x_s \in \mathcal X_s$,
there is a parameter vector $p \in \mathcal P$ with $S v(x_s,p) = 0$.
Then,
\begin{equation}
\begin{aligned}
\mathcal X_s = \left\lbrace x \in \Real^m \mid \exists p \in \mathcal P: S v(x,p) = 0 \right\rbrace.
\end{aligned}
\end{equation}
However, for the case $m>1$, when continuation methods are not suitable, 
there are at present no general methods to compute $\mathcal X_s$ efficiently and reliably.

We present a method to address this problem that works for arbitrarily large state and
parameter spaces, does not need to compute steady state values explicitly and is computationally
efficient.
The method is able to compute reliable, though conservative outer bounds on the set 
$\mathcal X_s$ of all steady states.

In order to search for sets of steady states for a given parameter set $\mathcal P$,
we need means to test whether a candidate solution $\mathcal X_s$ obtained in such a search
is actually valid or not.
Such a test is readily formulated as a feasibility problem.
Moreover, we will see that the Lagrangian dual for this feasibility problem
allows to certify given regions in state space as not containing a steady state
for any parameter value from the set $\mathcal P$.
We then develop an algorithm that uses this information to construct
outer bounds on the region $\mathcal X_s$ of all steady states.

In this paper, we consider only hyperrectangles for the sets $\mathcal X_s$
and $\mathcal P$ in state and parameter space.
An extension to more general convex polytopes is in principle 
easy from the theoretical perspective, but 
it requires a much more elaborate implementation on the practical side.

\section{Feasibility of steady state regions}
\label{sec:feasibility}

\subsection{Feasibility problem and semidefinite relaxation}

The problem of testing whether a given hyperrectangle $\mathcal X_s$ in state space contains steady states
of the system~\eqref{eq:ode}, for some parameter values in a given hyperrectangle $\mathcal P$ in parameter
space, can be formulated as the following feasibility problem:
\begin{equation}
\label{eq:primal-problem}
\textnormal(P):\ \left\lbrace\ 
\begin{aligned}
\textnormal{find} && x\in\Real^n,&\ p\in\Real^m & & \\
\textnormal{s.t.} && S v(x,p) &= 0 &  \\
&& p_{j,min} \leq p_j &\leq p_{j,max} & & j=1,\dotsc,m \\
&& x_{i,min} \leq x_i &\leq x_{i,max} & & i=1,\dotsc,n.
\end{aligned}\right.
\end{equation}

The same problem appears in the context of parameter identification in a recent paper
by \cite{KuepferSau2007}.
They developped a method that uses an infeasibility certificate for the 
problem~\eqref{eq:primal-problem} to exclude regions in parameter space from the
identification procedure, given a set of steady state measurements.
In this section, we take their approach to find an infeasibility certificate
for problem~\eqref{eq:primal-problem},
but give more details about the underlying mathematical techniques.

Relaxing the feasiblity problem~\eqref{eq:primal-problem} to a semidefinite
program \citep{VandenbergheBoy1996} ensures computational efficiency.
The applied relaxation is based on a quadratic representation
of a multivariate polynomial of arbitrary degree \citep{Parrilo2003a}.
In the first step, we construct a vector $\xi$ containing monomials that occur in
the reaction flux vector $v(x,p)$.
In the special case where no single reaction has more than two reagents, a starting
point for the construction of $\xi$ is
\begin{equation*}
\begin{aligned}
\xi\T = (1, p_1, \dotsc, p_m, x_1, \dotsc, x_n, p_1 x_1, \dotsc, p_j x_i, \dotsc, p_m x_n),
\end{aligned}
\end{equation*}
which can usually be reduced by eliminating components that are not required to represent
the reaction fluxes.
We define $k$ such that $\xi\in\Real^k$.
Note that this approach is not limited to second order reaction networks. 
In more general cases, one has to extend the vector $\xi$ by monomials that are products
of several state variables.

Using the vector $\xi$, the elements of the flux vector $v(x,p)$ can be expressed as
\begin{equation}
\label{eq:flux-reformulation}
\begin{aligned}
v_j(x,p) = \xi\T V_j \xi,\quad j=1,\dotsc,m,
\end{aligned}
\end{equation}
where $V_j \in \mathcal S^k$ is a constant symmetric matrix.
The choice of $V_j$ is generally not unique, as an expression of the form
$p_j x_i x_k$ can be decomposed as either $(p_j x_i) (x_k)$ or
$(p_j x_k) (x_i)$. 
This fact may be used to introduce additional equality constraints in the relaxed 
problem~\eqref{eq:relaxed-primal-problem}, but we will neglect this for simplicity of notation.

Using~\eqref{eq:flux-reformulation}, the system~\eqref{eq:ode} can be written as
\begin{equation}
\begin{aligned}
\dot x_i = \xi\T Q_i \xi,\quad i=1,\dotsc,n,
\end{aligned}
\end{equation}
where $Q_i = \sum_{j=1}^m S_{ij} V_j \in \mathcal S^k$ are constant symmetric matrices.

The original feasibility problem~\eqref{eq:primal-problem} is thus equivalent
to the problem
\begin{equation}
\begin{aligned}
\textnormal{find} && \xi \in\Real^k & & \\
\textnormal{s.t.} && \xi\T Q_i \xi &= 0 & i=1,\dotsc,n  \\
&& B \xi &\geq 0 \\
&& \xi_1 &= 1,
\end{aligned}
\end{equation}
where the matrix $B\in\Real^{(2k-2)\times k}$ is constructed to cover the inequality
constraints in~\eqref{eq:primal-problem},
e.g.\ the constraint $p_{1,min} \leq p_1 \leq p_{1,max}$ is represented as
\begin{equation*}
\begin{aligned}
\begin{pmatrix}
-p_{1,min} & 1 & 0 & \dotso & 0 \\
p_{1,max} & -1 & 0 & \dotso & 0
\end{pmatrix} \xi \boldsymbol\geq 0.
\end{aligned}
\end{equation*}
Corresponding constraints for higher order monomials in $\xi$ are obtained easily
as $p_{j,min} x_{i,min} \leq p_j x_i \leq p_{j,max} x_{i,max}$ and have to be
included in the matrix $B$.

A relaxation to a semidefinite program is found by setting $X = \xi \xi\T$.
The resulting non-convex constraint $\operatorname{rank} X = 1$ is omitted in
the relaxation.
Instead, several consequences of how $X$ is defined, namely $X_{11} = 1$
and $X \succcurlyeq 0$, are used as convex constraints.
The relaxed version of the original feasibility problem~\eqref{eq:primal-problem}
is thus obtained as
\begin{equation}
\label{eq:relaxed-primal-problem}
\textnormal(RP):\ \left\lbrace\ 
\begin{aligned}
\textnormal{find} && X\in \mathcal S^k & & \\
\textnormal{s.t.} && \tr(Q_i X) &= 0 & \quad i =1,\dotsc,n \\
&& \tr(e_1 e_1\T X) &= 1 & \\
&& BXe_1 &\boldsymbol\geq 0 & \\
&& BXB\T &\boldsymbol\geq 0 & \\
&& X &\succcurlyeq 0, & 
\end{aligned}\right.
\end{equation}
where $e_1 = (1,0,\dotsc,0)\T \in \Real^k$.

The basic relationship between the original problem~\eqref{eq:primal-problem} and
the relaxed problem~\eqref{eq:relaxed-primal-problem} is that if the original problem
is feasible, then the relaxed problem is also feasible.
Thus, the relaxation allows to certify a region in state space as infeasible
for steady states, as we will see when going to the Lagrange dual problem.

\subsection{Infeasibility certificates from the dual problem}

The Lagrange dual problem can be used to certify infeasibility of the
primal problem \eqref{eq:relaxed-primal-problem}.
First, the Lagrangian function $L$ is constructed for the primal problem.
We obtain
\begin{equation*}
\begin{aligned}
& L(X,\lambda_1,\lambda_2,\lambda_3,\nu) = -\lambda_1\T B X e_1 - \tr(\lambda_2\T B X B\T) \\
& \quad - \tr(\lambda_3\T X) + \sum_{i=1}^n \nu_i \tr(Q_i X) + \nu_{n+1}(\tr(e_1 e_1\T X)-1),
\end{aligned}
\end{equation*}
where $\lambda_1\in\Real^{2k-2}$, $\lambda_2\in\mathcal S^{2k-2}$, $\lambda_3\in\mathcal S^k$
and $\nu\in\Real^{n+1}$.
Using the cyclic property of the trace operator, i.e.\ $\tr(ABC) = \tr(BCA) = \tr(CAB)$, we
rewrite
\begin{equation*}
\begin{aligned}
\tr(\lambda_2\T B X B\T) = \tr(B\T \lambda_2\T B X)
\end{aligned}
\end{equation*}
and
\begin{equation*}
\begin{aligned}
\lambda_1\T B X e_1 &= \tr(e_1 \frac{\lambda_1\T}{2} B X) + \tr(e_1\T \frac{\lambda_1}{2} B\T X) \\
&= \tr((e_1 \frac{\lambda_1\T}{2} B+e_1\T \frac{\lambda_1}{2} B\T) X).
\end{aligned}
\end{equation*}
The second reformulation has also the advantage of providing a symmetric multiplier for $X$,
which is more efficient from the computational side.

Based on the Lagrangian $L$, the dual problem is obtained as
\begin{equation*}
\begin{aligned}
&\max\ \inf_{X \in \mathcal S^k} L(X,\lambda_1,\lambda_2,\lambda_3,\nu) \\
&\begin{aligned}
\textnormal{s.t.}&&\lambda_1 \boldsymbol\geq 0,\ 
\lambda_2 \boldsymbol\geq 0,\ 
\lambda_3 &\succcurlyeq 0,
\end{aligned}
\end{aligned}
\end{equation*}

which is equivalent to
\begin{equation}
\label{eq:dual-problem}
\textnormal(D):\left\lbrace
\begin{aligned}
&\max\quad \nu_{n+1} \\
&\begin{aligned}
\textnormal{s.t.}&&\ B\T \lambda_2 B + e_1 \lambda_1\T B + B\T \lambda_1 e_1\T &  \\
&& + \lambda_3 + \sum_{i=1}^n \nu_i Q_i + \nu_{n+1} e_1 e_1\T &= 0  \\
&&\lambda_1 \boldsymbol\geq 0,\ 
\lambda_2 \boldsymbol\geq 0,\ 
\lambda_3 &\succcurlyeq 0.
\end{aligned}
\end{aligned}\right.
\end{equation}

It is a standard procedure in convex optimisation to use the dual problem
in order to find a certificate that guarantees infeasibility of the
primal problem \citep{BoydVan2004}. 
For the problem at hand, this principle is formulated in the following theorem.
\begin{theorem}
\label{theo:infeasible}
If the dual problem \eqref{eq:dual-problem} has a feasible solution where
$\nu_{n+1}>0$,
then the primal problem \eqref{eq:primal-problem} is infeasible.
\end{theorem}
\begin{proof}
Note that the constraints of the dual problem~\eqref{eq:dual-problem} are homogenous in
the free variables: if $(\lambda_1^\prime,\lambda_2^\prime,\lambda_3^\prime,\nu^\prime)$
is feasible, then also 
$(\alpha\lambda_1^\prime,\alpha\lambda_2^\prime,\alpha\lambda_3^\prime,\alpha\nu^\prime)$
with any $\alpha \geq 0$ is feasible.
In particular, choosing all free variables to be zero is always a feasible solution of
the dual problem~\eqref{eq:dual-problem}.

Let $d^\ast$ be the optimal value of the dual problem~\eqref{eq:dual-problem}.
By the previous argument, it is clear that either $d^\ast = 0$ or $d^\ast = \infty$.
Under the assumption made in the theorem, we have $d^\ast = \infty$.

To the primal feasibility problem~\eqref{eq:relaxed-primal-problem},
we can associate a minimization
problem with zero objective function and the same constraints as in~\eqref{eq:relaxed-primal-problem}.
Let $p^\ast$ be the optimal value of this minimization problem.
We have $p^\ast = 0$, if the primal problem~\eqref{eq:relaxed-primal-problem} is
feasible, and $p^\ast = \infty$ otherwise.
Weak duality of semidefinite programs \citep{VandenbergheBoy1996} assures
that $d^\ast \leq p^\ast$.
In particular, $d^\ast = \infty$ implies $p^\ast = \infty$, and the primal
problem~\eqref{eq:relaxed-primal-problem} as well as the original feasibility
problem~\eqref{eq:primal-problem} are both infeasible.\hfill $\Box$
\end{proof}

Theorem~\ref{theo:infeasible} sets the basis for our further considerations.

\section{Bounding feasible steady states}
\label{sec:bounding}

In this section, we present an approach to find bounds on the steady
state region $\mathcal X_s$, based on the results obtained in the previous
section.
As basic additional requirement, we assume that some upper and lower bounds on
steady states are already known previously by other means.
Let these bounds be given by
\begin{equation}
\label{eq:pre-bounds}
\begin{aligned}
x_{i,lower} \leq x_i \leq x_{i,upper},\quad i=1,\dotsc,n.
\end{aligned}
\end{equation}
In biochemical reaction networks, such bounds can often be obtained from
mass conservation relations, as done for the examples in Section~\ref{sec:examples}.
Also, it is often possible to show positive invariance of a sufficiently large
compact set in state space for the system~\eqref{eq:ode}.
These bounds may be very loose though, and the main objective of our method
is to tighten them as far as possible.

To this end, we use a bisection algorithm that finds the maximum ranges
$[x_{j,lower},x_{j,min}]$ and $[x_{j,max},x_{j,upper}]$ for which infeasibility
can be proven via Theorem~\ref{theo:infeasible}. The algorithm iterates
over $j=1,\dotsc,n$, while the steady state values $x_i$ for $i \neq j$ are 
assumed to be located within the interval given by inequality~\eqref{eq:pre-bounds}.

We give the bisection algorithm in pseudocode for computing the lower bound
$x_{1,min}$. The computation of the upper bound $x_{1,max}$ works in essentially 
the same way, with some obvious modifications.

\begin{alg}[Lower bound maximization by bisection]
\label{alg:bisection}\hfill \\[-0.5cm]
\begin{ttfamily}
\begin{tabbing}
\hspace*{0.3cm} \= \makebox[0cm][l]{up\_guess <- $x_{1,upper}$}\hspace*{0.2cm} \= \hspace*{0.3cm} \= \\
\> lo\_guess <- $x_{1,lower}$ \\
\> next\_$x_1$ <- $x_{1,upper}$ \\
\> while (up\_guess - lo\_guess)$\,\geq\,$tolerance \\
\>\> use constraint $x_{1,lower} \leq x_1 \leq \mbox{}$next\_$x_1$ \\
\>\> solve semidefinite program $(D)$ \\
\>\> if $d^\ast = \infty$ \\
\>\>\> lo\_guess <- next\_$x_1$ \\
\>\>\> increase next\_$x_1$ by $\frac{1}{2}$(up\_guess - next\_$x_1$) \\
\>\> else \\
\>\>\> up\_guess <- next\_$x_1$ \\
\>\>\> decrease next\_$x_1$ by $\frac{1}{2}$(next\_$x_1$ - lo\_guess) \\
\>\> endif \\
\> endwhile \\
\> $x_{1,min}$ <- lo\_guess
\end{tabbing}
\end{ttfamily}
\end{alg}

Due to the availability of efficient solvers for semidefinite programs and
the use of bisection to maximize the interval that is certified as infeasible,
Algorithm~\ref{alg:bisection} can run considerably fast on standard desktop
computers, as we will see in the examples discussed in the following section.

In our analysis method, Algorithm~\ref{alg:bisection} is run for all state variables,
and as both maximization of the lower bound and minimization of the upper bound
of the steady state values.
Its output is a hyperrectangle in state space containing
all steady states for the assumed parameter ranges.
This is a relevant information for the global sensitivity analysis of a biochemical
reaction network, as it allows to discriminate concentration
values that are highly affected by the assumed parameter variations from
others that are less affected.
Moreover, by repeating the computation for different parameter ranges,
it is also possible to assess the influence of individual parameters
on steady state concentrations, which is closer related to classical,
local sensitivity analysis.

\section{Examples}
\label{sec:examples}

\subsection{A simple conversion reaction}

As first example, we consider a simple conversion reaction where the region of
steady states for a given parameter box can be computed analytically.
Consider the reaction network
\begin{equation*}
A\ \underset{k_2}{\overset{k_1}{\rightleftarrows}}\ B.
\end{equation*}
Denote the concentrations of $A$ and $B$ as $a$ and $b$, respectively.
There is a conservation relation $a(t)+b(t) = a_0$, so the system can
be modelled by one differential equation
\begin{equation}
\label{eq:conversion-reaction}
\begin{aligned}
\dot a = k_2 (a_0 - a) - k_1 a.
\end{aligned}
\end{equation}
Furthermore, there is a unique steady state $a_s$ for all parameter values, given
by
\begin{equation*}
\begin{aligned}
a_s = \frac{k_2 a_0}{k_1 + k_2}.
\end{aligned}
\end{equation*}
From the conservation relation, we have the loose bound $0 \leq a_s \leq a_0$
which is valid for all parameter values.
Assume now that $a_0 = 1$ is fixed, and let the other parameters vary in a box
$k_1,k_2\in[k_{min},k_{max}]$. Then, the steady state varies in the interval
\begin{equation*}
\begin{aligned}
a_s \in \left[\frac{k_{min}}{k_{max} + k_{min}},\frac{k_{max}}{k_{min} + k_{max}}\right].
\end{aligned}
\end{equation*}
In the specific case where $k_{min}=1$ and $k_{max}=2$, the steady state interval
is $a_s \in [\frac{1}{3},\frac{2}{3}]$.
Our algorithm is able to compute numerically exact bounds in these cases.
For a numerical precision of $10^{-6}$, computation time is a few seconds
on a standard desktop computer.

\subsection{An enzymatic cycle}

As a more complex example, where the steady state region for a given
parameter box cannot be computed analytically, we consider an enzymatic cycle.
These cycles appear very frequently in cellular reaction networks,
in particular in the form of phosphorylation/de\-phospho\-ry\-lation cycles \citep{ShacterCho1984}.
An enzymatic cycle as encountered in covalent modification of proteins \citep{GoldbeterKos1981}
is typically described by the reaction network
\begin{equation}
\label{eq:enzymatic-cycle}
\begin{aligned}
E + A\ &\underset{k_2}{\overset{k_1}{\rightleftarrows}}\ C_1\\
C_1\ &{\overset{k_3}{\rightarrow}}\ E + A^\ast \\
P + A^\ast\ &\underset{k_5}{\overset{k_4}{\rightleftarrows}}\ C_2\\
C_2\ &{\overset{k_6}{\rightarrow}}\ P + A.
\end{aligned}
\end{equation}

There are three conservation relations
\begin{equation*}
\begin{aligned}
{}[A]+[A^\ast]+[C_1]+[C_2] &= A_0 \\
[E] + [C_1] &= E_0 \\
[P] + [C_2] &= P_0.
\end{aligned}
\end{equation*}
Denoting $a = [A^\ast]$, $c_1 = [C_1]$ and $c_2 = [C_2]$ and using the law
of mass action, the reaction flux vector is given by
\begin{equation*}
\begin{aligned}
v = \begin{pmatrix}
k_1 (A_0 - a -c_1 - c_2)(E_0-c_1) \\
k_2 c_1 \\
k_3 c_1 \\
k_4 (P_0 - c_2) a \\
k_5 c_2 \\
k_6 c_2
\end{pmatrix}.
\end{aligned}
\end{equation*}
Due to the conservation relations, we only need to use three differential equations in the
model, which is given by
\begin{equation}
\begin{aligned}
\frac{d}{dt} (a,\ c_1,\ c_2)\T =
\begin{pmatrix}
0 & 0 & 1 & -1 & 1 & 0 \\
1 & -1 & -1 & 0 & 0 & 0 \\
0 & 0 & 0 & 1 & -1 & -1
\end{pmatrix} v.
\end{aligned}
\end{equation}
For the sensitivity analysis,
the parameters $k_1$ and $k_4$ as well as the total concentrations $A_0$, $E_0$ and $P_0$
are assumed to be fixed at $k_1 = 10^{5}$, $k_4 = 5\cdot 10^{4}$, $A_0 = 1$ and $E_0 = P_0 = 0.01$.
The other parameters are assumed to be variable parameters, with variations around their
nominal values
$k_{2,nom} = k_{5,nom} = 1$ and $k_{3,nom} = k_{6,nom} = 10^3$.

From the conservation relations and invariance of the positive orthant
we have the steady state bounds
\begin{equation*}
\begin{aligned}
0 \leq a \leq A_0,\quad
0 \leq c_1 \leq E_0,\quad
0 \leq c_2 \leq P_0,
\end{aligned}
\end{equation*}
which are valid for any parameter values.

We have applied the proposed analysis method to find tighter bounds
on possible steady state values, comparing three different regions in which parameters
of the enzymatic cycle are allowed to vary.
The three different regions are given by $\mathcal P_1$, $\mathcal P_2$ and $\mathcal P_3$,
where $\mathcal P_1,\mathcal P_2, \mathcal P_3 \subset \Real^4$ and
\begin{itemize}
\item $(k_2,k_3,k_5,k_6)\in\mathcal P_1 \Leftrightarrow 0.98\, k_{i,nom} \leq k_i \leq 1.02\, k_{i,nom}$,
corresponding to parameter variations of up to 2\%,
\item $(k_2,k_3,k_5,k_6)\in\mathcal P_2 \Leftrightarrow 0.9\, k_{i,nom} \leq k_i \leq 1.1\, k_{i,nom}$,
corresponding to parameter variations of up to 10\%, and
\item $(k_2,k_3,k_5,k_6)\in\mathcal P_3 \Leftrightarrow 0.5\, k_{i,nom} \leq k_i \leq 2\, k_{i,nom}$,
corresponding to up to 2--fold parameter variations,
\end{itemize}
with $i=2,3,5,6$ in all three cases.

The dual problem $(D)$ has been constructed by using
\begin{equation}
\begin{aligned}
\xi\T = (1,\ k_2,\ k_3,\ k_5,\ k_6,\ a,\ c_1,\ c_2),
\end{aligned}
\end{equation}
and deriving appropriate matrices $Q_i$, $B$, for the steady state equations and the
constraints, respectively.
Algorithm~\ref{alg:bisection} was then used to compute bounds on the
steady state concentrations.
We compare these results to an estimate for the region of steady state concentrations
obtained by Monte--Carlo tests.
The results are shown in Figure~\ref{fig:enzymatic-cycle-results}.
The average computation time to obtain the feasible intervals for all three state
variables and one parameter region was about 25 seconds.
The Monte--Carlo tests done to produce the figures took consistently about 20 \% more computation time,
where 1000 parameter points were used for each test.
However, for a reliable evaluation by Monte--Carlo methods, much more points
should be used, which would increase computation time significantly.

As can be seen from the figure, our approach is able to find tight intervals for the steady state
values of the individual concentrations.
However, the results also highlight the limitations of using hyperrectangles if the
steady state values are highly correlated.

\begin{figure}
\begin{center}
\subfigure[Parameter uncertainty region $\mathcal P_1$ (2\% variation)]{%
\includegraphics[width=0.48\linewidth]{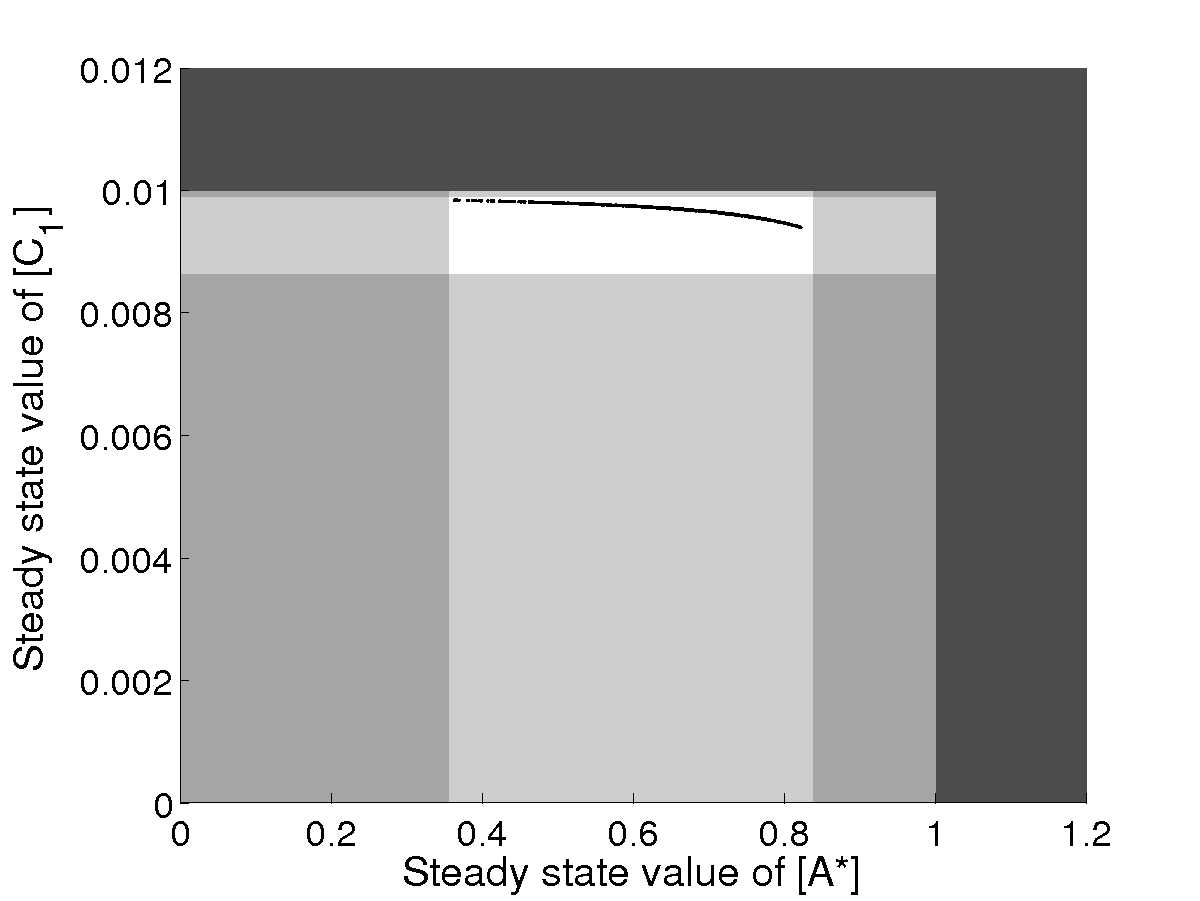}}
\subfigure[Parameter uncertainty region $\mathcal P_2$ (10\% variation)]{%
\includegraphics[width=0.48\linewidth]{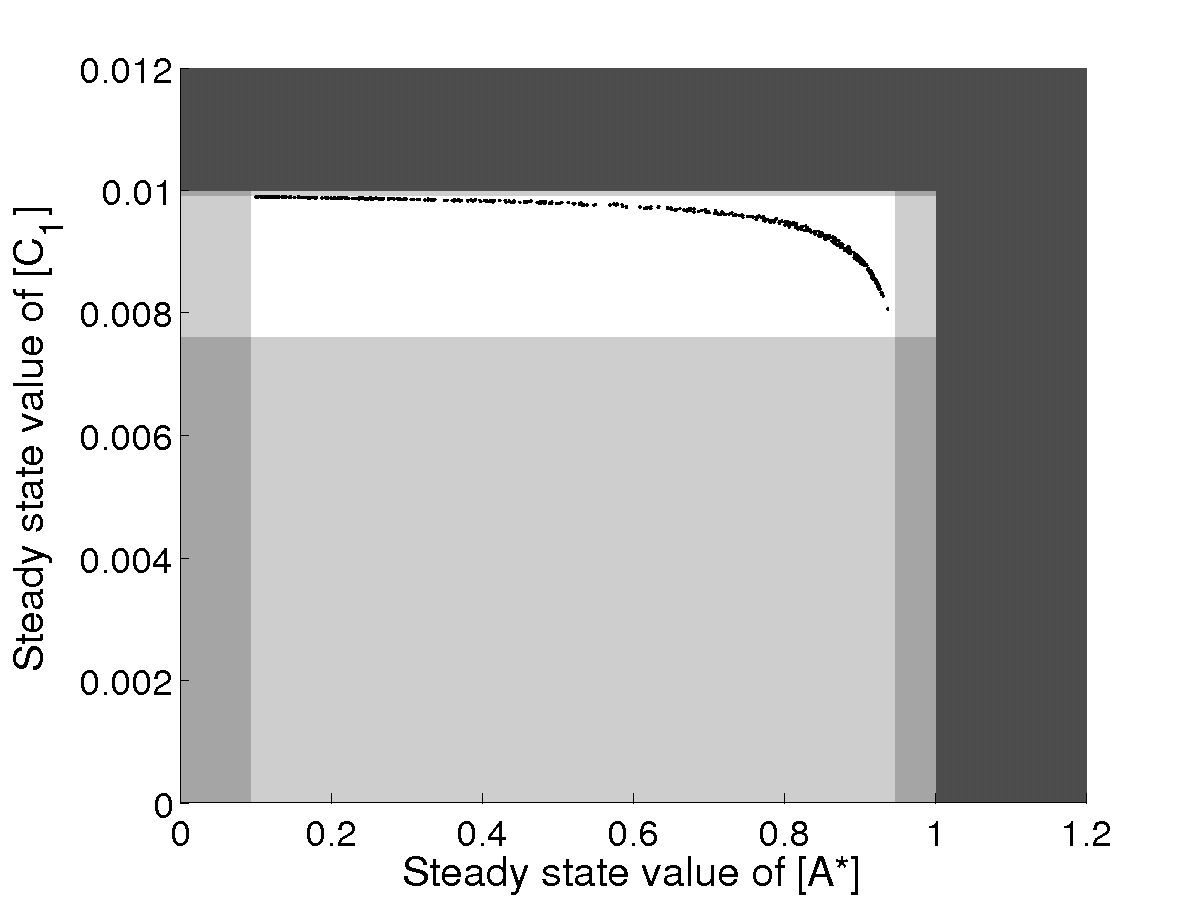}}\\
\subfigure[Parameter uncertainty region $\mathcal P_3$ (2--fold variation)]{%
\includegraphics[width=0.48\linewidth]{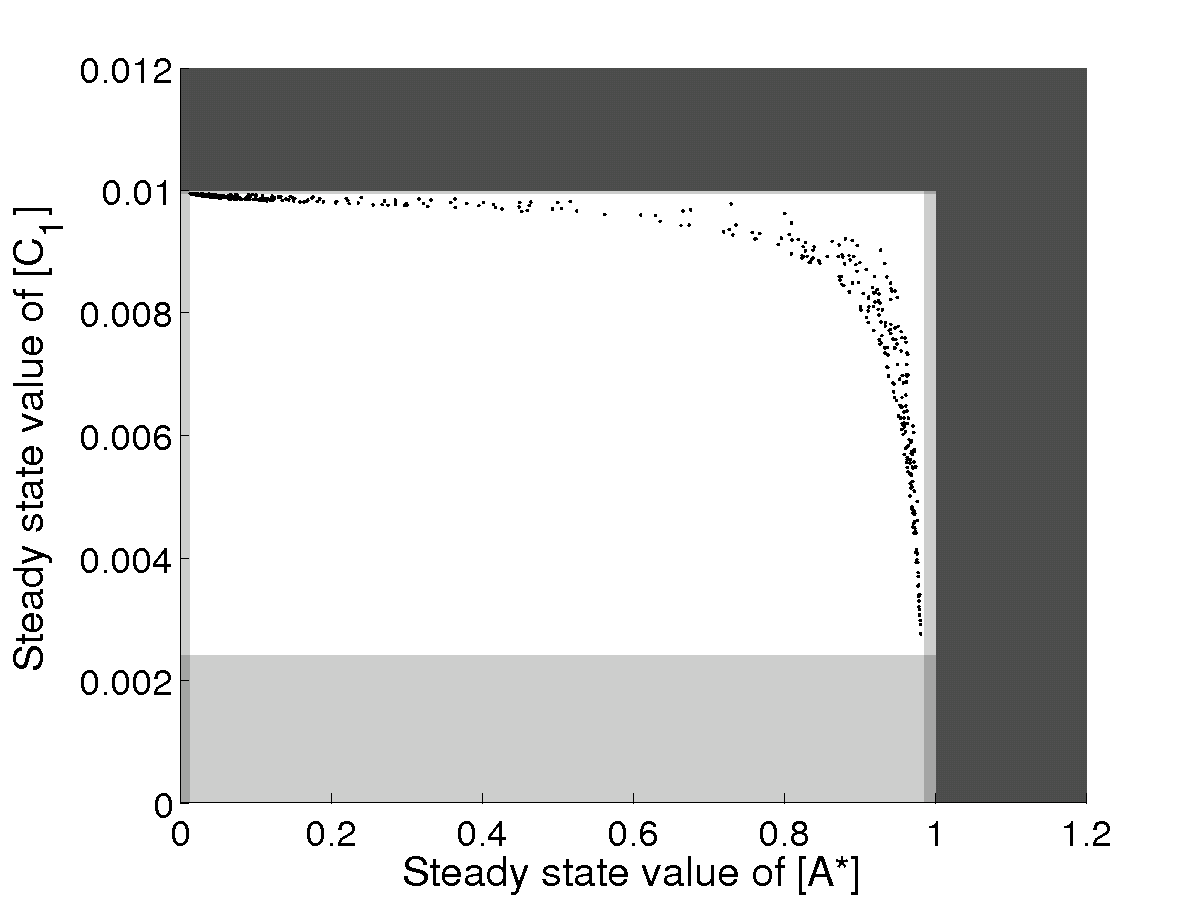}}
\end{center}
\caption{Feasible steady states for the enzymatic cycle with three different parameter
regions, comparison of reliable bounds obtained with Algorithm~\ref{alg:bisection} and
Monte--Carlo estimates. Light gray regions have been certified infeasible by
Algorithm~\ref{alg:bisection}. Black dots are steady state values obtained from Monte-Carlo tests.
Dark gray regions are known to be infeasible from conservation relations.}
\label{fig:enzymatic-cycle-results}
\end{figure}

Our analysis also yields a biochemical interpretation, related to the property of
ultrasensitivity.
The concept of ultrasensitivity is quite important for biochemical reaction networks,
in particular for those that constitute cellular signal transduction pathways
\citep{LevineKue2007}.
Shortly, ultrasensitivity means that a small variation in a control variable
has a relatively large effect on an output variable, whereas for increasing
variations in the control variable, the range of the output variable will be
considerably less increasing (see also Figure~\ref{fig:ultrasensitivity}).
Thus, ultrasensitivity is an inherently non-linear and non-local property.
For the enzymatic cycle,
a variation of only 2\% in parameters already allows the steady state value of
$[A^\ast]$ to vary over almost half of the interval given from the conservation relation,
and with an allowable parameter variation of 10\% the steady state value
of $[A^\ast]$ can span nearly the whole interval.
This is a clear indication of the ultrasensitivity which is typical for the enzymatic cycle
\citep{GoldbeterKos1981}.

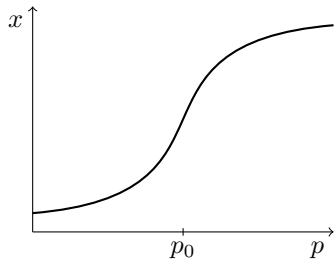
\begin{figure}[bt]
\begin{center}
\begin{tikzpicture}[scale=0.5]
\draw[->] (0,0) -- (8,0) node[below left] {$p$};
\draw[->] (0,0) -- (0,6) node[below left] {$x$};
\draw[thick] (0,0.5) .. controls +(6,0.5) and +(-6,-0.5) .. (8,5.5);
\draw (4,0.1) -- (4,-0.1) node[pos=0.5,below] {$p_0$};
\end{tikzpicture}
\end{center}
\caption{Illustration of ultrasensitivity. Response of an output variable $x$ to a control
variable $p$. The response is ultrasensitive, with high sensitivity around the nominal
value $p_0$ and considerable less sensitivity for other values.}
\label{fig:ultrasensitivity}
\end{figure}

In addition, our results show that the steady state value of $[C_1]$, the concentration
of the intermediate enzyme--substrate complex, is not ultrasensitive, because
its value spans a large interval only for large parameter variations.
Similar results hold for $[C_2]$.

\section{Conclusions}

We have studied the problem of computing the region of all steady states of biochemical
reaction networks, provided that parameters are allowed to vary within a known region.
This is an important problem in sensitivity analysis of reaction networks.
Our approach is based on formulating a feasibility problem to check whether a candidate
region in state space actually contains steady states.
This feasibility problem is relaxed to a semidefinite program, and its Lagrangian
dual provides certificates of infeasibility of a candidate region in state space.
These certificates can be used to efficiently minimize the estimate of the known feasible region
in state space by a bisection algorithm.

We have applied our sensitivity analysis to two simple example networks.
For the first example, our algorithm is able to compute numerically exact bounds, which
could be verified from the analytical solution.
In the second example, we compared the bounds obtained from our algorithm to steady state
values obtained through Monte--Carlo tests. 
In this example, our approach was more efficient computationally than Monte--Carlo tests.
Also, it gives guaranteed bounds on the steady state values, which
cannot be achieved by randomized methods such as Monte--Carlo tests.
Based on the premise that we are working
with hyperrectangles only, the obtained bounds are fairly tight.
The second example also shows that our approach is able to confirm ultrasensitivity
of the Goldbeter--Koshland switch.

In summary, our approach is a reliable and computationally efficient method to estimate the range
of possible steady state variations due to multiple simultaneous parameter variations in biochemical
reaction networks, and thus provides a valuable tool for global sensitivity analysis.

\bibliography{/home/waldherr/Forschung/Referenzen.bib}

\end{document}